%% file: wordle.tex
\documentclass[11pt]{article}

\usepackage[margin=1in]{geometry}
\usepackage[numbers,sort&compress]{natbib}
\usepackage{amsmath}
\usepackage{amssymb}
\usepackage{amsthm}
\usepackage{algorithm}
\usepackage{algorithmic}
\usepackage{booktabs}
\usepackage{enumitem}
\usepackage{multicol}
\usepackage{xcolor}
\usepackage{soul}
\usepackage{graphicx}
\usepackage[small]{caption}
\usepackage{microtype}
\usepackage{hyperref}
\usepackage{doi}
\usepackage{tikz}
\usepackage{eufrak,bbold}
\hypersetup{
    colorlinks=true,
    linkcolor=black,
    citecolor=black,
    filecolor=black,
    urlcolor=[HTML]{0088CC},
    pdftitle={},
}

\theoremstyle{definition}

\newtheorem{thm}{Theorem}

\newtheorem{prop}{Proposition}[section]

\newtheorem{lem}[prop]{Lemma}
\newtheorem{cor}[prop]{Corollary}

\newtheorem{fact}[prop]{Fact}
\newtheorem{rem}[prop]{Remark}

\newtheorem{dfn}[prop]{Definition}

\newcommand{\dft}[1]{\textbf{\textit{#1}}}

\newcommand{\abs}[1]{\left|#1\right|}

\newcommand{\set}[1]{\left\{#1\right\}}
\newcommand{\sucht}{\,\middle|\,}

\newcommand{\Z}{\mathbf{Z}}

\newcommand{\WI}{\mathrm{WI}}
\newcommand{\WinningStrategy}{\mathrm{WinningStrategy}}
\newcommand{\filter}{\mathrm{filter}}
\newcommand{\solvable}{\mathrm{solvable}}
\newcommand{\sol}{\mathrm{sol}}
\newcommand{\true}{\mathbf{true}}
\newcommand{\false}{\mathbf{false}}

\renewcommand{\th}{{}^{\mathrm{th}}}

\DeclareMathOperator{\poly}{poly}

\title{Finding a Winning Strategy for Wordle is NP-complete}

\author{
  Will Rosenbaum\\
  Amherst College\\
  wrosenbaum@amherst.edu
}      

\date{\today}

\begin{document}

\maketitle

\thispagestyle{empty}

\begin{abstract}
  In this paper, we give a formal definition of the popular word-guessing game Wordle. We show that, in general, determining if a given Wordle instance admits a winning strategy is NP-complete. We also show that given a Wordle instance of size $N$, finding a winning strategy that uses $g$ guesses in the worst case (if any) can be found in time $N^{O(g)}$. 
\end{abstract}

\input{intro}
\input{preliminaries}

\input{algorithm}
\input{hardness}

\input{conclusion}

%\input{conclusions}

\urlstyle{same}
\bibliographystyle{plainnat}
\bibliography{wordle}

\end{document}

%% file: intro.tex
\section{Introduction}
\label{sec:intro}

The game Wordle is an elegant word-guessing game released by Josh Wardle in October 2021~\cite{Victor2022-wordle}. The premise and gameplay are simple: a player has six chances to guess an (unknown) target five letter word. For each letter in the guessed word, the player receives feedback of the following form:
\begin{enumerate}
\item the letter appears in the same position in the target word,
\item the letter appears in the target word, but in a different position
\item the letter does not appear in the target word.
\end{enumerate}
Using this feedback, a player can adaptively choose a sequence of up to six words. The player wins if they correctly guess the target word within six guesses, and the player loses otherwise. Since Wordle's release, it has become immensely popular. As of the end of January, 2022, Wordle had millions of daily players (myself included)~\cite{Tracy2022-new-york}.

In this paper, we give a formal description of (a generalization of) Wordle. Specifically, a \dft{Wordle instance} specifies lists of possible target words (all of the same length, $d$), a dictionary of allowed guesses, and the maximum number of guesses $g$ the player can make. The decision problem $\WinningStrategy$ is to determine if there is a (deterministic) strategy that will result in a player always winning---i.e., following the strategy will always result in the player correctly guessing any given target word in $g$ or fewer tries.

Our main result is to show that $\WinningStrategy$ is NP-complete (Theorem~\ref{thm:hardness}). We prove NP-hardness via a reduction from the minimum dominating set (MDS) problem (Lemma~\ref{lem:hardness}). To show that $\WinningStrategy$ is in NP, we describe how a Wordle strategy can be encoded as a \emph{strategy tree} (Definition~\ref{dfn:strategy-tree}), whose size is polynomial in the Wordle instance size. The maximum number of guesses required to win corresponds to the tree's depth, and the validity of a purported strategy tree can be verified in polynomial time. Thus, the strategy tree affords a complete and sound certificate for a winning Wordle strategy.

Additionally, we show that for any Wordle instance $\WI$ with fixed game length $g$, an (optimal) winning strategy (if any) can be found in $\abs{\WI}^{O(g)}$ time. In particular, for fixed constant $g$, a winning strategy can be found in polynomial time.

\subsection{Related Work}

The very recent (indpendent) work of Lokshtanov and Subercaseaux~\cite{Lokshtanov2022-wordle} also establishes the NP-hardness finding a winning strategy for Wordle. Their definition of Wordle is slightly different from our Definition~\ref{dfn:wordle-instance}. In their model, guessed words must be chosen from the set of possible target words, whereas our definition allows for a strictly larger set of allowable guess words,\footnote{In our notation, Lokshtanov and Subercaseaux assume $D = W$.} although they do suggest the generalization we consider as Open Problem~3. The NP-hardness proof of Lokshtanov and Subercaseaux is strictly stronger than our Lemma~\ref{lem:hardness}, as their argument implies that $\WinningStrategy$ remains NP-hard when words have length at most $5$, and that $\WinningStrategy$ is $W[2]$-hard when parameterized by game length.

On the other hand, our Theorem~\ref{thm:hardness} establishes that $\WinningStrategy \in \mathrm{NP}$, which is listed as Open~Problem~1 by Lokshtanov and Subercaseaux. Our proof of $\WinningStrategy \in \mathrm{NP}$ also extends to the context of their Open~Problem~3, where membership in the dictionary $D$ of allowed queries may be defined implicitly via a finite automaton (or any polynomial-time computable function). Specifically, our argument shows that in this case, $\WinningStrategy$ remains in NP, even though the dictionary size $\abs{D}$ may be exponential in the instance size. We refer the reader to Lokshtanov and Subercaseaux~\cite{Lokshtanov2022-wordle} for further discussion and related work.

%% file: preliminaries.tex
\section{Preliminaries}
\label{sec:preliminaries}

In this section, we formally define a Wordle Instance and the associated task of finding a winning strategy for an instance.

\begin{dfn}
  \label{dfn:wordle-instance}
  A \dft{Wordle instance} $\WI = (\Sigma, d, D, W, g)$ consists of:
  \begin{itemize}
  \item a finite \dft{alphabet} $\Sigma$,
  \item a \dft{dimension} $d \in \Z^+$,
  \item a \dft{dictionary} $D \subseteq \Sigma^d$,
  \item a \dft{word list} $W \subseteq D$,
  \item a \dft{game length} $g \in \Z^+$.
  \end{itemize}
\end{dfn}

We note the distinction between the dictionary $D$ and the word list $W$. The player can use any word $u \in D$ as a guess during gameplay, while the target word will always be chosen from $W$.

\begin{dfn}
  \label{dfn:wordle-oracle}
  Given a word $w \in W$ and a query $u \in D$, the \dft{Wordle oracle} $\Omega_w$ returns the response $r = \Omega_w(u) \in \set{0, 1, 2}^d$, defined as follows:
  \begin{equation}
    \label{eqn:wordle-oracle}
    r_i =
    \begin{cases}
      2 &\text{if } w_i = u_i\\
      1 &\text{if } \abs{\set{j \leq i \sucht u_i = u_j, u_j \neq w_j}} \leq \abs{\set{j \sucht w_j = u_i, w_j \neq u_j}}.\\
      0 &\text{otherwise}
    \end{cases}
  \end{equation}
\end{dfn}

Intuitively, $r_i = 1$ indicates that the letter $u_i$ appears in $w$, but not in position $i$. The somewhat complicated second condition in~(\ref{eqn:wordle-oracle}) is to deal with case that $w$ or $u$ contains repeated letters. If a letter $x$ is repeated $k$ times in $w$, then up to the first $k$ instances of $x$ in $u$ can get a corresponding response $1$, while subsequent instances of $x$ in $u$ will get a $O$ responses. For example, if $w = \texttt{HELLOOO}$, we would have $\Omega_w(\texttt{OOOOHHH}) = \texttt{1110100}$.

The goal of Wordle is given a Wordle instance $\WI$ and query access to a Wordle oracle $\Omega$, find the unique word $w \in W$ for which $\Omega = \Omega_w$ using as few queries to $\Omega$ as possible.

\begin{dfn}
  \label{dfn:strategy}
  Given a Wordle instance $\WI$, a \dft{strategy} defines for any sequence of queries and responses $u^1, r^1, u^2, r^2, \ldots, u^{k-1}, r^{k-1}$ a next query
  \[
  u^k = q(u^1, r^1, u^2, r^2, \ldots, u^{k-1}, r^{k-1}) \in D.
  \]
  We say that the strategy \dft{succeeds in round $k$} if $\Omega(u^k) = \texttt{22}\cdots\texttt{2}$, hence we can conclude that $\Omega = \Omega_w$ for $w = u^k$. We say that a strategy is a \dft{winning strategy} if for all $w \in W$, the strategy succeeds in at most $g$ (the game length) rounds.
\end{dfn}

There are several natural computational questions that arise from the definitions above. Given a Wordle instance $\WI$, is there a winning strategy? Can a winning strategy be found efficiently? Towards answering these questions, we first give a recursive characterization of Wordle instances with winning strategies.

\begin{prop}
  \label{prop:winning}
  Let $\WI = (\Sigma, d, D, W, g)$ be Wordle instance. Given a query word $u \in D$ and target word $w \in W$, let
  \begin{equation}
    \label{eqn:consistent-set}
    C_{w}(u) = \set{v \in W \sucht \Omega_w(u) = \Omega_v(u)}
  \end{equation}
  Then $\WI$ admits a winning strategy if and only if either $\abs{W} = 1$ and $g \geq 1$, or there exists a query word $u \in D$ such that for every target word $w \in W$ the instance
  \[
  \WI' = (\Sigma, d, D, C_{w}(u), g-1)
  \]
  has a winning strategy.
\end{prop}
\begin{proof}
  It is clear that $\WI$ with $g = 1$ has a winning strategy if and only if $\abs{W} = 1$, so we consider the case $g > 1$.

  First suppose there exists $u \in D$ such that for every $w \in W$, $\WI'$ as above has a winning strategy. Then a winning strategy for $\WI$ can be performed by using $u$ as its first query, and emulating the winning strategy for the resulting $\WI'$ for its remaining queries.

  Conversely, suppose $\WI$ has a winning strategy. Then each $\WI'$ as above has a winning strategy formed by simply following the winning strategy of $\WI$ after the first query (that resulted in the instance $\WI'$).
\end{proof}

\subsection{Strategy Trees}

Here, we give a more refined characterization of strategies in terms of a \emph{strategy tree}. The basic idea is as follows. Throughout a seqeunce of interactions with a Wordle oracle, the player maintains a set $C \subseteq W$ of words consistent with the query-response pairs comprising the player's interaction with the oracle. Specifically, given any interaction $(u^1, r^1, u^2, r^2, \ldots, u^k, r^k)$ with an oracle $\Omega$, we define the sequence of $C_0 = W, C_1, C_2, \ldots, C_k$ inductively by
\[
C_j = \set{w \sucht \forall i \leq j,\ \Omega_w(u^i) = r^i}.
\]
We call a set $C = C_j$ a \dft{consistent set} for the interaction with a Wordle oracle $\Omega$. Informally, a consistent set contains all possible words $w \in W$ that are consistent with $\Omega$'s responses so far.

Given a consistent set $C$ and a query $u \in D$, we say that elements $w, w' \in C$ are \dft{$u$-equivalent} and write $w \sim_{u} w'$ if $\Omega_{w}(u) = \Omega_{w'}(u)$. Clearly, $\sim_u$ is an equivalence relation. We let $C(u)$ denote the \dft{partition induced by $u$}, i.e., the partition of $C$ into equivalence classes according to $\sim_u$. Thus, $C(u)$ is a partition of $C$ into at most $3^d$ parts, corresponding to the possible responses $\Omega_w(u) \in \set{0, 1, 2}^d$ for $w \in C$. The following definition formalizes a condition under which a query $u$ provides useful information to the player.

\begin{dfn}
  \label{dfn:informative}
  Given a consistent set $C$ and query word $u \in D$, we say that $u$ is \dft{informative} if $\abs{C(u)} > 1$. That is, the induced partition $C(u)$ contains multiple parts. Otherwise, $u$ is \dft{uninformative}.
\end{dfn}

Intuitively, an informative query is a query whose response narrows down the set of consistent words. If $C_0 = W, C_1, \ldots, C_k$ is the sequence of consistent words with an interaction $u^1, r^1, \ldots, u^k, r^k$, then query $u^i$ is informative if and only if $C_i \neq C_{i-1}$.

We observe that if $\abs{C} > 1$, then every query $u \in C$ is informative. In particular, $u \in C$ is the unique word $w$ for which $\Omega_w(u) = \texttt{22}\cdots\texttt{2}$. Further, if a strategy $S$ ever makes an uninformative query, we can devise a more efficient strategy $S'$ by simply omitting the uninformative queries in $S$. Thus, for the remainder of the paper, assume without loss of generality that all strategies only make informative queries. We call such strategies \dft{informative strategies}.

\begin{dfn}
  \label{dfn:strategy-tree}
  Given an (informative) strategy $S$, we define the \dft{strategy tree} $T(S) = (V, E)$ to be the rooted tree with vertex set $V \subseteq 2^{W} \times D$ where each vertex $v = (C, u) \in V$ consists of a consistent set $C$ and (informative) query $u$ for $C$. Each edge $e \in E$ is labeled with a response $r \in \set{0,1,2}^d$. If $v = (C, u)$ has child $v' = (C', u')$ and the edge $(v, v')$ has label $r$, then $C' = \set{w \in C \sucht \Omega_w(u) = r}$. The root of $T$ is labeled with $(W, u^1)$, where $u^1$ is the first query according to $S$. Finally, the leaves of $T$ are labeled $(\set{w}, w)$.
\end{dfn}

Given a strategy $S$, the interpration of the strategy tree $T = T(S)$ is as follows. An execution of $S$ begins at the root of $T$, $(W, u^1)$. After making the first query $u^1$, the child $(C_1, u^2)$ of $(W, u^1)$ incident to the edge labeled $r^1 = \Omega(u^1)$ is selected. $C_1$ is the subset of words in $C_0$ that are consistent with the query/response $(u^1, r^1)$, and $u^2$ is the next query made according to $S$. Continuing in this way, each path from root to leaf in $T$ corresponds to a query/response sequence in an execution of $S$. Since the leaves are labeled with singleton sets, $(\set{w}, w)$, the strategy $S$ succeeds in the round after guessing $w$. With this interpretation of strategy trees, the following observations are clear.

\begin{lem}
  \label{lem:strategy-tree}
  Let $\WI = (\Sigma, d, D, W, g)$ be a Wordle instance, let $S$ be a strategy, and $T = T(S)$ the corresponding strategy tree. Then:
  \begin{enumerate}
  \item $S$ succeeds in $k$ rounds against Wordle oracle $\Omega_w$ if and only if the leaf $(\set{w}, w)$ is at depth $k - 1$ in $T$.
  \item $S$ is a winning strategy if and only if $T$ has depth $g - 1$.
  \item $S$ is an informative strategy if and only if every internal vertex in $T$ has at least $2$ children.
  \end{enumerate}
\end{lem}

Item~2 above allows us to completely characterize winning Wordle strategies in terms of their associated strategy trees. Item~3, in turn, allows us to bound the size of strategy trees for informative strategies.

\begin{lem}
  \label{lem:leaves}
  Suppose $T = (V, E)$ is a rooted tree with $N$ leaves such that each internal vertex has at least $2$ children. Then $\abs{V} \leq 2 N - 1$.
\end{lem}
\begin{proof}
  Recall that the height $h(v)$ is a vertex is defined to be
  \[
  h(v) =
  \begin{cases}
    0 &\text{if } v \text{ is a leaf}\\
    1 + \max \set{h(w) \sucht w \text{ is a child of } v} &\text{otherwise.}
  \end{cases}
  \]
  A straightforward induction argument on $h(v)$ shows that every vertex $v$ with $N_v$ descendent leaves has at most $2 N_v - 1$ descendants (including $v$ itself). The lemma follows by taking $v$ to be the root.
\end{proof}

\begin{cor}
  \label{cor:strategy-tree-size}
  Let $\WI = (\Sigma, d, D, W, g)$ be a wordle instance and $S$ an informative strategy. Then $T(S)$ has at most $2 \abs{W} - 1$ vertices. In particular, $T(S)$ can be expressed in size $\poly(\abs{\WI})$.
\end{cor}
\begin{proof}
  Let $T = T(S) = (V, E)$ be the strategy tree for $S$. First observe that for each $w \in W$, there is a unique leaf $(\set{w}, w)$ in $V$. Indeed, since the children of each vertex $(C, u)$ correspond to a partition of $C$, sets $C, C'$ corresponding to vertices $(C, u), (C', u) \in V$ are disjoint unless $(C, u)$ is an ancestor of $(C', u')$, or vice versa. Therefore, $T$ has $\abs{W}$ leaves. Further, since $S$ is an informative strategy, all internal vertices in $T$ have at least two children. Thus, by Lemma~\ref{lem:leaves}, $T$ has at most $2 \abs{W} + 1$ vertices. Since each vertex label has size $O(\abs{\WI})$, $T$ has size $O(\abs{\WI}^2)$.
\end{proof}

\begin{rem}
  \label{rem:certificate-size}
  In our calculation of the description size of $T$ above, note that each vertex of $T$ is labeled with a subset of $W$ and an element in $D$. Thus, the total size $O(\abs{W} (\abs{W} + \log \abs{D})$. In particular, this is polynomial in the size of $W$ even if $\abs{D}$ is exponential in the size of $W$.
\end{rem}

%% file: algorithm.tex
\section{A Generic Algorithm}
\label{sec:algorithm}

Here, we describe a simple procedure that given a Wordle instance $\WI = (\Sigma, d, D, W, g)$ determines if $\WI$ has a winning solution. The procedure exhaustively searches all strategies requireing at most $g$ interactions until a suitable strategy (or none) is found. While this procedure is impractical, it runs in time $O(\abs{\WI}^{O(g)})$ using space $O(g \abs{\WI})$. Thus, for games of fixed constant game length $g$, $\WI$ can be solved in polynomial time and linear space.

\subsection{Basic Algorithm}

The goal of this section is to prove the following theorem.

\begin{thm}
  \label{thm:algorithm}
  There exists an algorithm that, given a Wordle instance $\WI = (\Sigma, d, D, W, g)$, determines if $\WI$ admits a winning strategy in $O(\abs{\WI}^{O(g)})$ time using space $O(g \abs{\WI})$.
\end{thm}

%% \begin{cor}
%%   \label{cor:pspace}
%%   The problem of finding a winning strategy for Wordle is in PSPACE.
%% \end{cor}

Towards proving Theorem~\ref{thm:algorithm}, we first describe a simple subroutine, $\filter(C, u, \Omega)$, that given a set of words $C$, a query $u$, and a Wordle oracle $\Omega$, returns the subset $A \subseteq C$ of words that are consistent with $\Omega(u)$.

\begin{algorithm}[H]
  \caption{$\filter(C, u, \Omega)$. Given a set $C$ of words, a query word $u$, and a Worldle oracle $\Omega$, return the subset of words in $C$ that are consistent with $\Omega(u)$.
    \label{alg:filter}}
  \begin{algorithmic}[1]
    \STATE $A \leftarrow \varnothing$
    \STATE $r \leftarrow \Omega(u)$
    \FORALL{$v \in C$}
    \IF{$\Omega_v(u) = r$}
    \STATE $A \leftarrow A \cup \set{v}$
    \ENDIF
    \ENDFOR
    \RETURN $A$
  \end{algorithmic}
\end{algorithm}

We state the main properties of Algorithm~\ref{alg:filter} in the lemma below.

\begin{lem}
  \label{lem:filter}
  Algorithm~\ref{alg:filter} uses 1 query to $\Omega$. Assuming $C \subseteq W$ and $A$ are represented as lists of words in $\Sigma^d$, the algorithm can be implemented in time $O(\abs{\WI} \log d)$ time and $O(\abs{\WI})$ space.
\end{lem}
\begin{proof}
  The first assertion of the lemma is clear, as $\Omega$ is only invoked in Line~2. We observe that computing $\Omega_v(u)$ can be computed in time $O(d \log (d + \abs{\Sigma})$. The \texttt{2} entries of $\Omega_v(u)$ can be found in time $O(d \log \abs{\Sigma})$ by simply compring $v_i$ and $u_i$ for $i = 1, 2, \ldots, d$. The $\texttt{1}$-entries, can be found in $O(d \log d \log \abs{\Sigma}$ time be, e.g., by sorting the non-$2$ entries of $u$ and $v$ in order to apply the second case of Equation~(\ref{eqn:wordle-oracle}). Finally, this process is iterated $\abs{C}$ times to form $A$, giving a total running time of $O(\abs{C} d \log d \log \abs{\Sigma}) = O(\abs{\WI})$.
\end{proof}

We now present a recursive method that starting from a given state---i.e., a collection $C$ of target words in $W$ that are consistent with previous queries---determines if the game can be won with $\ell$ remaining queries. The basic idea is to leverage the recursive description of Wordle instances with winning strategies described in Proposition~\ref{prop:winning}. By the proposition, in order to determine if an instance has a winning strategy starting from state $C$ in $\ell$ steps, it suffices to determine if there exists a query $u$ such that for each possible response and its corresponding state $C' \subseteq C$, the instance has a winning strategy of length $\ell - 1$ starting from $C'$. The base case occurs when either $\abs{C} = 1$ and $\ell \geq 1$ or $\ell = 1$.

Algorithm~\ref{alg:solvable} gives pseudocode for a method, $\solvable$, that implements the recursive procedure described above. For $\ell, \abs{C} > 1$, the method iterates over choices of queries $u \in D$ (lines \ref{ln:query-start}--\ref{ln:query-end}). For each query, the method iterates over all consistent target words $w \in C$ (lines \ref{ln:target-start}--\ref{ln:target-end}) and determines the set $A \subseteq C$ of words consistent with the result $\Omega_w(u)$. A recursive call to $\solvable(A, D, \ell-1)$ determines if the resulting state admits a winning strategy in $\ell - 1$ rounds. If a $u$ is found such that all resulting states admit winning strategies, then the value $\true$ is returned (line~\ref{ln:success}). Otherwise, if no such $u$ is found, the value $\false$ is returned.

\begin{algorithm}
  \caption{$\solvable(C, D, \ell)$
    \label{alg:solvable}}
    \begin{multicols}{2}
  \begin{algorithmic}[1]
    \IF{$\ell = 1$ and $\abs{C} > 1$}\label{ln:base-start}
    \RETURN $\mathbf{false}$
    \ENDIF
    \IF{$\abs{C} = 1$}
    \RETURN $\mathbf{true}$
    \ENDIF\label{ln:base-end}
    \FORALL{$u \in D$}\label{ln:query-start}
    \STATE $\sol \leftarrow \true$
    \FORALL{$w \in C$}\label{ln:target-start}
    \STATE $A \leftarrow \filter(C, u, w)$
    \IF{$\neg\solvable(A, D, \ell - 1)$}
    \STATE $\sol \leftarrow \false$
    \STATE $\mathbf{break}$
    \ELSE
    \STATE $\sol \leftarrow \true$\label{ln:success}
    \ENDIF
    \ENDFOR\label{ln:target-end}
    \IF{$\sol$}
    \RETURN $\true$
    \ENDIF
    \ENDFOR\label{ln:query-end}
    \RETURN $\false$
  \end{algorithmic}
  \end{multicols}
\end{algorithm}

\begin{lem}
  \label{lem:solvable-correctness}
  Let $\WI = (\Sigma, d, D, W, g)$ be a Wordle instance, $C \subseteq W$ a set of consistent words, and $\ell \leq g$. Then $\solvable(C, D, \ell)$ returns $\true$ if and only if $\WI' = (\Sigma, d, D, C, \ell)$ admits a winning strategy. In particular, $\WI$ admits a winning strategy if and only if $\solvable(W, D, g)$ returns $\true$.
\end{lem}
\begin{proof}
  We argue by induction on $\ell$. The base case $\ell = 1$ follows from the first case of Proposition~\ref{prop:winning} and lines~\ref{ln:base-start}--\ref{ln:base-end}. For the inductive step, suppose the lemma holds for $\ell - 1$. A call to $\solvable(C, D, \ell)$ returns $\mathbf{true}$ at line~\ref{ln:success} iff there exists $u \in D$ such that for all $w \in C$, $\solvable(A, D, \ell-1)$ returns $\mathrm{true}$, where $A = \set{v \in C \sucht \Omega_v(u) = \Omega_w(u)}$. By the inductive hypothesis, this occurs iff all instances are solvable from state $A$ within $\ell - 1$ rounds. Finally, by Proposition~\ref{prop:winning}, this occurs iff $\WI'$ is solvable, as desired.
\end{proof}

%% \subsection{Modifications and Extensions}

%% - every game has a solution of length $\abs{W}$

%% - finding optimal solution length is in PSPACE

%% - given $g$ and a distribution over $W$, we can find a strategy that maximizes the probability of success?

%% file: hardness.tex
\section{NP-Completeness }
\label{sec:hardness}

In this section, we prove the following theorem.

\begin{thm}
  \label{thm:hardness}
  Determining if a Wordle instance has a winning strategy is NP-complete.
\end{thm}

The proof of NP-hardness follows from a reduction from the minimum dominating set problem (MDS). NP-completeness follows because every (winning, informative) strategy admits a description---its strategy tree---whose size is polynomial in in its size and whose validity can be verified in polynomial time.

Recall that given a graph $G = (V, E)$, a \dft{dominating set} is a subset of vertices $U \subseteq V$ such that every vertex is adjacent to a vertex in $U$. For a given parameter $K$, MDS asks whether $G$ has a dominating set of size (at most) $K$. MDS was shown to be NP-complete in~\cite{Garey1979-computers}. Towards proving Theorem~\ref{thm:hardness}, we will require a slight modification of Garey and Johnson's result.

\begin{fact}
  \label{fact:mds}
  Let $G$ be a graph that is promised to have a minimum dominating set of even cardinality, and let $K$ be an even number. Then it is NP-hard to determine if $G$ has a dominating set of size at most $K$ or if $G$'s minimum dominating set has size at least $K + 2$.
\end{fact}

Fact~\ref{fact:mds} follows immediately from Garey and Johnson's hardness result by considering graphs $G$ of the form $G = G' \sqcup G'$. That is $G$ is a disjoint union of two copies of some graph $G'$. Since there are no edges between to the disjoint copies of $G'$, $G$ has a dominating set of size $K = 2 K'$ if and only if $G'$ has a dominating set of size $K'$.

\begin{lem}
  \label{lem:hardness}
  The problem of determining if a Wordle instance has a winning strategy is NP-hard.
\end{lem}
\begin{proof}
  Let $G = (V, E)$ be a graph with $V = [n] = \set{1, 2, \ldots, n}$ and $E = \set{e_1, e_2, \ldots, e_m}$, and fix any even value of $K < n$. Given $G$, we will construct a Wordle instance $\WI = (\Sigma, d, D, W, g)$ such that $\WI$ has a winning strategy of length $K + 1$ if and only if $G$ has a dominating of size at most $K + 1$. If $G$ is promised to have an even-cardinality MDS (as in Fact~\ref{fact:mds}), then determining if $\WI$ has a winning strategy of length $K+1$ certifies that $G$ has an MDS of size at most $K$.

  The idea of our reduction is that words in $D = W$ correspond to vertices in $G$, and that given $w, u \in W$, $\Omega_w(u)$ gives a ``positive'' (non-zero) answer if and only if $w$ and $u$ correspond to neighbors in $G$. More formally, we set:
  \begin{itemize}
  \item $\Sigma = V \cup E$
  \item $d = m = \abs{E}$
  \item $D = W = \set{w_1, w_2, \ldots, w_n} \subseteq \Sigma^m$ where $w_i$ is defined as
    \begin{equation}
      w_{ij} =
      \begin{cases}
        e_j &\text{if } i \text{ is incident with } e_j\\
        i &\text{otherwise}.
      \end{cases}
    \end{equation}
  \item $g = K + 1$.
  \end{itemize}
  Since each $w_i$ corresponds to vertex $i$ in $G$, we will refer to these elements interchangeably as words and vertices. Clearly, $\WI(G, K)$ as defined above can be computed in polynomial time. The crux of our argument is the following claim about the vector returned by a Wordle oracle.
  \begin{description}
  \item[Claim.] For each index $i$, let $\mathbf{e}_i$ denote the $i\th$ standard unit vector---i.e., $\mathbf{e}_{ij}$ is $1$ for $i = j$ and $0$ otherwise. Then
    \begin{equation}
      \label{eqn:oracle}
      \Omega_{j}(i) =
      \begin{cases}
        0 = \texttt{00}\cdots\texttt{0} &\text{if } (i, j) \notin E\\
        \texttt{22}\cdots\texttt{2} &\text{if } i = j\\
        2 \textbf{e}_k &\text{if } e_k = (i, j) \in E.
      \end{cases}
    \end{equation}
  \item[Proof of Claim.] If $(i, j) \notin E$, then $w_i$ and $w_j$ do not share any symbols in common because all symbols in $w_i$ are either $i$ or $(i, j') \in E$. Thus, $\Omega_{j}(i)$ returns all \texttt{0}s. For the third case, suppose $(i, j) = e_k \in E$. Then $w_{ik} = w_{jk} = e_k$, while $w_i$ and $w_j$ differ on all other characters, as claimed.
  \end{description}
  Returning to the main proof, we first argue the $\Leftarrow$ direction. To this end assume that $G$ has a dominating set $U = \set{u_1, u_2, \ldots, u_K}$ of size $K$, and fix any $j \in W$. Since $U$ is a dominating set, either $j \in U$ or there is some $u_i \in U$ such that $(u_i, j) \in E$. Consider the strategy that queries $u_1, u_2, \ldots$ until the first query $u_i$ returns a nontivial (i.e., nonzero) answer. If $u_i = j$, then the strategy succeeds. Otherwise, if $(u_i, j) \in E$, then $\Omega_{j}(i)$ returns $2 \textbf{e}_k$ where $e_k = (u_i, j)$. The player then picks $w_j$ as their next choice, thus winning in at most $\abs{U} + 1 = K + 1$ rounds.

  For the $\Rightarrow$ direction, suppose $\WI(G, K)$ has a winning strategy $S$ of length $K+1$. 
  %$G$ has a minimum dominating set of size at least $K + 2$. Fix any (deterministic) strategy $S$, and
  Let $u_1$ be the first query made by $S$. Define $W_1 = \set{j \in [n] \sucht \Omega_j(u_1) = 0}$. By the previous claim, $W_1$ consists of all non-neighbors of $u_1$. Inductively define $u_i$ to be the $i\th$ query assuming all previous queries returned $0$, and define $W_i = \set{j \in W_{i-1} \sucht \Omega_{j}(u_i) = 0}$. Again, by the claim, $W_i$ consists of all vertices that are not neighbors of any vertex in $U_i = \set{u_1, u_2, \ldots, u_i}$. Since $S$ is a winning strategy, there is some $k \leq K+1$ for which $W_k = \varnothing$, for otherwise, choosing $w \in W_{K+1}$, we have $\Omega_{w}(u_i) = 0$ for all $i$, hence $S$ loses. Since $W_k = \varnothing$, $U_k$ correspondings to a dominating set in $G$ of size $k \leq K + 1$. By the promise that $G$ has an MDS of even cardinality, such a dominating set certifies that the minimum dominating set of $G$ has size at most $K$.
\end{proof}

To prove Theorem~\ref{thm:hardness}, we must additionally show that finding a winning Wordle strategy is in NP. We will show that every $\WI$ with a winning strategy admits a winning strategy that can be described in $O(\poly(\abs{\WI}))$-space and verified in $O(\poly(\abs{\WI}))$ time. To this end, we rely on the characterization of NP as the class of decision problems that admit efficiently verifiable proof systems (see, e.g.,~\cite{Goldreich2008-computational}, Definition~2.5 and Theorem~2.8).

\begin{proof}[Proof of Theorem~\ref{thm:hardness}]
  By Lemma~\ref{lem:hardness}, determining if a Wordle instance has a winning strategy is NP-hard. All that remains is to show that the problem is in NP. To this end, let $\WI = (\Sigma, d, D, W, g)$ be a Wordle instance.

  First suppose $\WI$ admits a winning strategy $S$. Then $\WI$ admits an \emph{informative} winning strategy $S'$ (Definition~\ref{dfn:informative}). By Corollary~\ref{cor:strategy-tree-size}, the strategy tree $T = T(S') = (V, E)$ can be encoded in $\poly(\abs{\WI})$ space. Moreover, the correctness of $T$ can be vefified in $\poly(\abs{\WI})$ time by applying the $\filter$ procedure to each edge $e \in E$ in order to verify that $T$ satisfies Definition~\ref{dfn:strategy-tree}.

  Conversely, suppose $\WI$ does not admit a winning strategy. Then given any description of a strategy tree $T$, either the validation of $T$ using $\filter$ on each edge will fail, or the tree $T$ has depth greater $g - 1$. Thus, any purported strategy tree $T$ will be rejected. Therefore, $\WinningStrategy$ is in NP, as desired.
\end{proof}

\begin{rem}
  \label{rem:dictionary}
  As noted in Remark~\ref{rem:certificate-size}, the size of the certificate (i.e., strategy tree) does not depend explicitly on $\abs{D}$. As suggested in Open Problem~3 in~\cite{Lokshtanov2022-wordle}, one could consider a setting in which $\abs{D}$ is definited implicitly as, say, the words accepted by a finite automaton or polynomial time algorithm. In this setting, $\abs{D}$ may be exponential in $\abs{\WI}$. Nonetheless, the argument above still yields that in this generalized setting, $\WinningStrategy$ remains in NP. All that is required is to add a step to the certificate verification procedure that checks that each query $u$ satisfies $u \in D$.
\end{rem}

%% file: conclusion.tex
\section{Conclusion and Questions}
\label{sec:conclusion}

We conclude with a few related open questions.

\begin{enumerate}
\item Does Wordle (with the fixed word list and dictionary in the game's original implementation) admit a winning strategy? Is there a practical winning strategy, say, that can be executed by a human player?
\item Consider a distributional variant of Wordle in which each target word $w$ is assigned a probability, $p(w)$---the probability with which $w$ is chosen. In this model, we can define two optimization problems:
  \begin{itemize}
  \item find a strategy that maximizes the winning probability with respect to $p$;
  \item find a strategy that minimizes the expected number of guesses needed to find the target word.
  \end{itemize}
  The first variation is NP-hard, as $\WinningStrategy$ is the restricted decision problem, ``Can the winning probability be made $1$?'' Is this variant in NP? Is the second variant NP-hard?
\end{enumerate}

%% Questions:

%% - Practical stategies?

%% - Still hard with additional restrictions? (E.g., finite alphabet)

%% - Approximate solutions?

%% - FPT Algorithms for $g$?